\documentclass[a4paper]{article}
\usepackage[utf8]{inputenc}
\usepackage[dvipsnames]{xcolor}
\usepackage[UKenglish]{babel}
\usepackage{amsmath,amsthm,amssymb}
\usepackage{mathtools}
\usepackage[margin=3.5cm]{geometry}
\usepackage{pgfplots}
\usepackage{booktabs}
\usepackage{longtable}
\usepackage{caption}
\usepackage[hidelinks]{hyperref}
\usepackage{url}
\usepackage{cleveref}

\newtheorem{theorem}{Theorem}
\newtheorem{lemma}[theorem]{Lemma}

\pgfplotsset{compat=1.17}

\newcommand{\mypar}[1]{\medskip

\noindent {\bf #1}}

\title{Abstract morphing using the Hausdorff distance\\and Voronoi diagrams\footnote{Research was funded by NWO TOP grant no.~612.001.651.}}
\author{Lex de Kogel \and Marc van Kreveld \and Jordi L. Vermeulen}
\date{}

\renewcommand{\phi}{\varphi}
\renewcommand{\dh}{d_H}
\newcommand{\dhd}{d_{\vec{H}}}
\newcommand{\dilate}[1]{#1^\oplus}

\newcommand{\norm}[1]{\left|#1\right|}
\newcommand{\seg}[1]{\overline{#1}}
\newcommand{\comp}[1]{\#C(#1)}
\DeclareMathOperator{\partition}{Par}

\definecolor{nicered}{rgb}{0.737,0.11,0.11}
\definecolor{niceblue}{rgb}{0.11,0.11,0.737}
\definecolor{nicegreen}{rgb}{0.11,0.737,0.11}

\pgfplotscreateplotcyclelist{customcycle}{
    {niceblue,mark=*},
    {nicegreen,mark=square*},
    {nicered,mark=diamond*}
}

\begin{document}

\maketitle

\begin{abstract}
This paper introduces two new abstract morphs for
two $2$-dimensional shapes. The intermediate shapes gradually reduce the Hausdorff distance to the goal shape and increase the Hausdorff distance to the initial shape. The morphs are conceptually simple and apply to shapes with multiple components and/or holes.
We prove some basic properties relating to continuity, containment, and area. Then we give an experimental analysis that includes the two new morphs and a recently introduced abstract morph that
is also based on the Hausdorff distance~\cite{van2022between}. We show results on the area and perimeter development throughout the morph, and also the number of components and holes. A visual comparison shows that one of the new morphs appears most attractive.
\end{abstract}

\section{Introduction}

Morphing, also referred to as shape interpolation, is the changing of a given shape into a target shape over time. Applications include animation and medical imaging.
Animation is often motivated by the film industry, where morphing can be used to create cartoons or visual effects. In medical imaging, the objective is a 3D reconstruction from cross-sections, such as those from MRI or CT scans. Reconstruction between two 2D slices is essentially 2D interpolation between shapes, which is a form of morphing.
We regard morphing itself as the change of one shape into another shape by a parameter, or, more precisely, a function from the interval $[0,1]$ to shapes in a space, such that the image at $0$ is the one input shape and the image at $1$ is the other input shape. It is often convenient to see the morphing parameter as time. In the rest of this paper, we will refer to the shape of the morph at any particular time value as an \emph{intermediate shape}. See \cref{fig:example-morph} for an example of two halfway shapes between polygons resembling a butterfly and a spider.

\begin{figure}[!t]
    \centering
    \begin{minipage}{0.25\textwidth}
        \includegraphics[width=\textwidth,clip=true,trim=0 0 0.1cm 0.3cm]{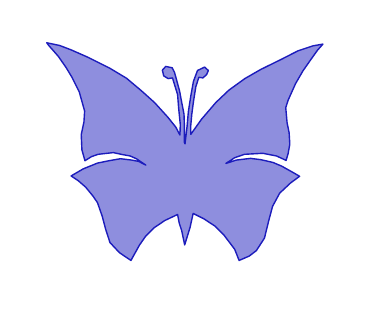}

\vspace*{-2mm}

        \includegraphics[width=\textwidth,clip=true,trim=0 0 0.1cm 0.3cm]{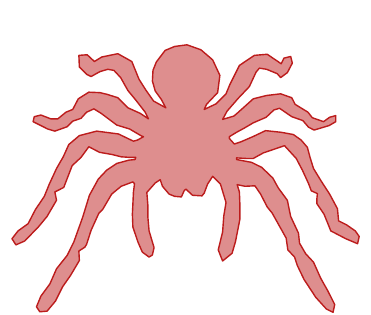}
    \end{minipage}%
    \begin{minipage}{0.75\textwidth}
        \includegraphics[width=0.5\textwidth]{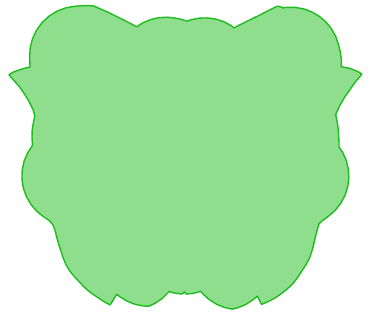}%
        \includegraphics[width=0.5\textwidth]{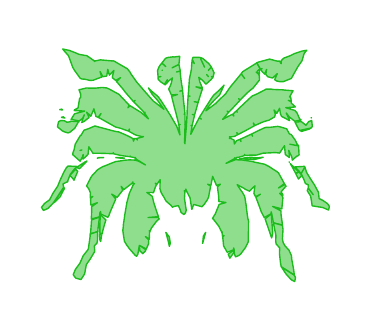}%
    \end{minipage}
    \caption{The intermediate shapes of two different morphing methods at time value \(1/2\) when morphing between the input shapes on the left. The middle shows the dilation morph (introduced in \cite{van2022between}), the right shows the Voronoi morph (introduced in this paper).}
    \label{fig:example-morph}
\end{figure}

The quality of a morph depends on the application. For medical imaging, the implied 3D reconstruction must be anatomically plausible. For morphing between two drawings of a cartoon character, the shapes in between must keep the dimensions of the limbs, for example. Furthermore, semantically meaningful features (nose, chin) should morph from their position in the one shape to their position in the other shape.

In this paper we concentrate on abstract morphing of shapes. A morphing task is abstract if there is no (semantic) reason to transform certain parts of a starting shape into certain parts of a goal shape.
In a recent paper,
van Kreveld et al.\ \cite{van2022between} presented a new type of abstract morphing based on the Hausdorff distance. It takes any two compact planar shapes \(A\) and \(B\) as input, and produces a morphed shape that interpolates smoothly between them. For a time value \(\alpha \in [0, 1]\), this morph is equal to \(A\) at \(\alpha = 0\) and to \(B\) at \(\alpha = 1\). For any value of $\alpha$ it has Hausdorff distance \(\alpha\) to \(A\) and Hausdorff distance \(1 - \alpha\) to \(B\), if the initial Hausdorff distance is $1$ (the input can be scaled to make this true without changing intermediate shapes). Morphs with this property are called \emph{Hausdorff morphs}~\cite{van2022between}. The Hausdorff morph introduced by van Kreveld et al.\ is based on Minkowski sums with a disk, and hence we refer to this specific one as the \emph{dilation morph}.

While the dilation morph has nice theoretical properties, in practice it will often grow intermediate shapes from \(A\) until \(\alpha = 1/2\), at which point the greatly dilated shape will shrink back towards~\(B\). For \(\alpha\) close to \(1/2\), the morphed shape typically resembles neither of the input shapes unless they already looked alike. We can see this in \cref{fig:example-morph}.

In this paper we present a new Hausdorff morph called \emph{Voronoi morph} that gives a more visually convincing morph, while maintaining many of the properties of the previous work. Our morph uses Voronoi diagrams to partition each input shape into regions with the same closest point on the other shape, and then scales and moves each such region to that closest point based on the value of \(\alpha\). We show that the Voronoi morph is also a Hausdorff morph. It interpolates smoothly between \(A\) and \(B\), but does not have the same problem of significantly increasing the area during the morph. We also present a variant called \emph{mixed morph} that reduces the problem of unnecessarily increasing the perimeter of the interpolated shape. It uses dilation and erosion to overcome some shortcomings of the Voronoi morph.

\mypar{Related work.}
The Hausdorff distance is a widely used distance metric that can be used for any two subsets of a space. It is a bottleneck measure: only a maximum distance determines the Hausdorff distance.
Efficient algorithms to compute the Hausdorff distance between two simple polygons or their higher-dimensional equivalents exist~\cite{alt95,alt03,atallah83}.
The Hausdorff distance is used in computer vision~\cite{dubuisson1994modified} and computer graphics~\cite{aspert2002mesh,cignoni1998metro} for template matching, and the computation of error between a model and its simplification.

Several algorithmic approaches to morphing have been described. Many of these are motivated by shape interpolation between slices (e.g.,~\cite{albu2008morphology,barequet04,barequet96,boissonnat88}, an overview can be found in \cite{barequet2009reconstruction}).
Other papers discuss morphing explicitly and not as an interpolation problem. 
Many of these results use compatible triangulations~\cite{gotsman01,LIU201860}, in particular those that avoid self-intersections.
It is beyond the scope of this paper to give a complete overview of morphing methods. For (not so recent) surveys of shape matching, interpolation, and correspondence, see~\cite{alt00,van2011survey}.
Our paper builds upon the morphing approach given in~\cite{van2022between}, which introduced Hausdorff morphs as a new technique for abstract morphing, and the dilation morph as a specific example of a Hausdorff morph.

Another shape similarity measure than the Hausdorff distance, the \emph{Fr\'echet} distance, can also be used to define a morph. In particular, \emph{locally correct Fr\'echet matchings} \cite{buchin2019locally} immediately imply a smooth transition of one shape outline into another, because they match all pairs of points on the two curves. Similar approaches were given in \cite{maheshwari2018approximating,rote2014lexicographic}. During the transition, however, the outline may be self-intersecting. This problem was addressed in~\cite{buchin2017frechet,chambers2011isotopic}.
A more important shortcoming of morphing using the Fr\'echet distance is that it is unclear how to morph between shapes with different numbers of components and holes.

Much of the commercial software for morphing applies to images, with or without additional human control. Other software is meant as toolkits for designers to design their own morphs, most notably Adobe After Effects.

\mypar{Our results.}
We introduce two new abstract morphs based on the
Hausdorff distance. They are---just like the dilation morph---conceptually simple and easy to implement if one has code for Minkowski sum and difference with a disk, Voronoi diagrams, and polygon intersection and union.
We examine basic properties of the two new morphs
and compare how they relate to the dilation morph.
In particular, we show that the Voronoi morph is a Hausdorff morph and that it is 1-Lipschitz continuous. We also show that for any morphing parameter (time), the Voronoi-morph intermediate shape is a subset of the mixed-morph intermediate shape, which in turn is a subset of the dilation-morph intermediate shape.

We then proceed with an extensive experimental analysis where we compare four basic quantities: area, perimeter, number of components, and number of holes. We show how these quantities develop throughout the three morphs. We also present visual results. As data we use simple drawings of animals, country outlines, and text (letters and whole words).

\section{Preliminaries}
Given two sets \(A\) and \(B\), we can define the directed Hausdorff distance from \(A\) to \(B\) as
\begin{equation*}\label{def:directed-hausdorff-distance}
    \dhd(A, B) \coloneqq \adjustlimits\sup_{a \in A} \inf_{b \in B} d(a, b) \text{,}
\end{equation*}
where \(d\) denotes the Euclidean distance. The \emph{undirected Hausdorff distance} between \(A\) and \(B\) is then defined as the maximum of both directed distances:
\begin{equation*}\label{def:hausdorff-distance}
    \dh(A, B) \coloneqq \max(\dhd(A, B),\; \dhd(B, A)) \text{.}
\end{equation*}

When \(A\) and \(B\) are closed sets, we can alternatively define the Hausdorff distance using Minkowski sums. Recall that the Minkowski sum \(A \oplus B\) is defined as
 $\{a + b~|~a \in A,~b \in B\}$;
the directed Hausdorff distance between \(A\) and \(B\) is then the smallest value \(r\) for which \(A \subseteq B \oplus D_r\), where \(D_r\) is a disk of radius \(r\).

Van Kreveld et al.~\cite{van2022between} then define a function that interpolates between two shapes in a Hausdorff sense: For any time parameter \(\alpha \in [0, 1]\), they define the dilation morph
\begin{equation*}\label{def:dilation-morph}
    S_\alpha(A, B) \coloneqq (A \oplus D_\alpha ) \cap (B \oplus D_{1 - \alpha}) \text{,}
\end{equation*}
and prove that this shape has Hausdorff distance \(\alpha\) to \(A\) and \(1 - \alpha\) to \(B\), and that it is the maximal shape with this property. Additionally, they show that this morph is 1-Lipschitz continuous: for two time parameters \(\alpha\) and \(\beta\), \(\dh(S_\alpha(A, B), S_\beta(A, B)) \leq \norm{\beta - \alpha}\). Note that we will omit the arguments $A,B\,$ when they are clear from context.

Structurally, it turns out that the intermediate shapes may have quadratic complexity, even when the input is two simple polygons with $n$ vertices each. For instance, if the input consists of a horizontal comb and an overlapping vertical comb, each with \(n/4\) prongs, \(S_\alpha\) will consist of \(\Omega(n^2)\) components for any \(\alpha \in (0, 1)\). In fact, this is not limited to the dilation morph: \emph{any} intermediate shape with Hausdorff distance \(\alpha\) to \(A\) and \(1 - \alpha\) to \(B\) will have \(\Omega(n^2)\) components~\cite{van2022between}, so every Hausdorff morph has this feature.

Note that both \(S_\alpha\) and the morphing methods described below change when we translate one of the input shapes. That is, if we write \(t + B\) to be the translation of \(B\) over a vector $t$, it is not true that \(S_\alpha(A, t + B) = \alpha * t + S_\alpha(A, B)\), as one might expect. This is because the positions of the input shapes are important both for the Hausdorff distance and for the shape of \(S_\alpha\). However, we can simply calculate \(S_\alpha(A, B)\) with \(A\) and \(B\) aligned in some way, and then generate \(S_\alpha(A, t + B)\) by explicitly translating \(S_\alpha(A, B)\) by \(\alpha * t\). Sensible alignment methods include aligning the centroids, maximising the overlap of the shapes, and minimising the Hausdorff distance.

\section{Voronoi morph}
As demonstrated in \cref{fig:example-morph}, one of the problems with the dilation morph is that the intermediate shapes tend to lose any resemblance to the input during the morphing process. The main reason for this is that the dilated shapes we are intersecting contain many points that do not influence the Hausdorff distance in any way, because they are not on the shortest path from a point on one shape to the closest point on the other. In other words, much of \(S_\alpha\) can be removed without changing the Hausdorff distance to and from the input. That said, there is no obvious ``correct'' way to determine which parts should be removed to obtain the greatest resemblance to the input.

We propose a morph in which we only take the points of \(S_\alpha\) that are on the shortest path between points in one input shape and the closest point on the other. Specifically, we only take the points where the ratio of distances to the one shape and the closest point on the other is \(\alpha : 1 - \alpha\). More formally, we define our new morph \(T_\alpha\) as follows:
\begin{equation*}\label{def:voronoi-morph}
    T_\alpha(A, B) \coloneqq \{a + \alpha(c(a, B) - a)~|~a \in A\} \cup \{b + (1 - \alpha)(c(b, A) - b)~|~b \in B\} \text{,}
\end{equation*}
where \(c(a, B)\) denotes the point on \(B\) closest to \(a\). In other words, we move each point in \(A\) closer to the closest point in \(B\) by a fraction \(\alpha\) of that distance, and each point in \(B\) closer to the closest point in \(A\) by a fraction \(1 - \alpha\), and take the union of those two shapes. If a point is equidistant to multiple points in the other shape, we include all options. We can prove that this morph has the desired Hausdorff distances to the input.

\begin{theorem}\label{thm:voronoi-morph-hausdorff-bounds}
    Let \(A\) and \(B\) be two compact sets in the plane with \(\dh(A, B) = 1\). Then for any \(0\leq \alpha \leq 1\), we have \(\dh(A, T_\alpha) = \alpha\) and \(\dh(B, T_\alpha) = 1 - \alpha\).
\end{theorem}
\begin{proof}
    We first show that \(\dh(A, T_\alpha) \leq \alpha\), and then show strict equality. The case for \(\dh(B, T_\alpha)\) is analogous and therefore omitted.

    By construction, any point \(a \in A\) has a point at distance \(\leq \alpha\) in \(T_\alpha\), showing that \(\dhd(A, T_\alpha) \leq \alpha\). Similarly, by construction, for each point \(b \in B\) there is a point \(t_b \in T_\alpha\) such that \(t_b = (1 - \alpha)(c(b, A) - b)\). As \(d(b, c(b, A)) \leq 1\), it must be the case that \(t_b\) has distance at most \(\alpha\) to \(c(b, A)\). It follows that all points in \(T_\alpha\) have distance at most \(\alpha\) to a point in \(A\), thereby showing that \(\dhd(T_\alpha, A) \leq \alpha\). As we have both \(\dhd(A, T_\alpha) \leq \alpha\) and \(\dhd(T_\alpha, A) \leq \alpha\), it follows that \(\dh(A, T_\alpha) \leq \alpha\).

    To show strict equality, assume the Hausdorff distance is realised by some point \(\hat{a} \in A\) with closest point \(\hat{b} \in B\), i.e., \(d(\hat{a}, \hat{b}) = 1\). By construction, there is a point \(\hat{t} \in T_\alpha\) at distance \(\alpha\) from \(\hat{a}\) and at distance \(1 - \alpha\) from \(\hat{b}\). As \(\hat{t}\) is the closest point to \(\hat{a}\) in \(T_\alpha\), we have \(\dh(A, T_\alpha) = \alpha\), and as \(\hat{b}\) is the closest point to \(\hat{t}\) in \(B\), we have \(\dh(B, T_\alpha) = 1 - \alpha\).
    If the Hausdorff distance is realised by a point on~\(B\), we use a symmetric argument.
\end{proof}

We can additionally show that the Voronoi morph, like the dilation morph, is 1-Lipschitz continuous in a Hausdorff sense:

\begin{lemma}\label{lem:voronoi-morph-lipschitz}
    Let \(\alpha, \beta \in [0, 1]\). Then \(\dh(T_\alpha, T_\beta) \leq \norm{\beta - \alpha}\).
\end{lemma}
\begin{proof}
    Let \(t_\alpha\) be any point on \(T_\alpha\). Assume without loss of generality that there is some \(a \in A\) such that \(t_\alpha = a + \alpha(c(a, B) - a)\) (the case for \(t_\alpha\) being included due to a point in \(B\) is analogous). Now consider the point \(t_\beta = a + \beta(c(a, B) - a)\): \(t_\alpha\) and \(t_\beta\) are on the same straight line segment between \(a\) and \(c(a, B)\), and have distance \(\norm{\beta - \alpha} \cdot \norm{c(a, B) - a}\) to each other. As \(\dh(A, B) = 1\), we know that \(\norm{c(a, B) - a} \leq 1\), and therefore that \(\norm{t_\beta - t_\alpha} \leq \norm{\beta - \alpha}\). This holds for any \(t_\alpha \in T_\alpha\), and the argument is symmetric for \(T_\beta\).
\end{proof}

\noindent
Note that this type of continuity implies that components of \(T_\alpha\) can only form or disappear by merging with or splitting from another component.

In addition to the Hausdorff distance-related properties, it is also interesting to study the general geometric and topological properties of \(T_\alpha\). We first show that the number of components \(\comp{T_\alpha}\) of \(T_\alpha\) does not change during the morph, except possibly at \(\alpha = 0\) and \(\alpha = 1\). We prove this for the case of polygonal input; the proof can likely be generalised, but the formalisation is somewhat tedious and not particularly interesting.

\begin{figure}
    \centering
    \includegraphics[page=1]{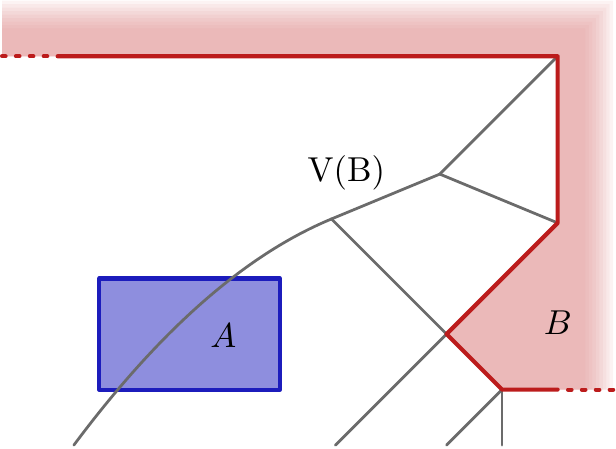}
    \hfill
    \includegraphics[page=2]{figures/voronoi-morph-example.pdf}
    \caption{On the left, \(A\) is partitioned by the Voronoi diagram \(V(B)\) of \(B\). On the right, each partitioned part of \(A\), shown in green, is scaled towards the closest point on \(B\) by a factor \(\alpha\).}
    \label{fig:voronoi-morph-example}
\end{figure}

Let \(V(A)\) be the Voronoi diagram of the vertices, open edges and the interior components of \(A\). We now define \(\partition(A, B)\) to be the input shape \(A\) partitioned into regions by \(V(B)\). Note that \(\partition(A, B)\) is a set of regions of \(A\) that each have the closest point of \(B\) on the same vertex, edge or face of \(B\). For some region \(P \in \partition(A, B)\), let \(P_\alpha\) be the region obtained by scaling \(P\) towards the site of the Voronoi cell of \(B\) it is in by a factor \(\alpha\). If this site is a vertex, we simply scale \(P\) uniformly towards it; if it is an edge, we scale it perpendicular to the supporting line of that edge; and if it is a face, it does not scale or move at all; see \cref{fig:voronoi-morph-example} for an illustration. Now let \(\partition_\alpha(A, B) \coloneqq \{P_\alpha~|~P \in \partition(A, B)\}\). Note that \(T_\alpha\) is the union of all elements of \(\partition_\alpha(A, B)\) and \(\partition_{1 - \alpha}(B, A)\).

\begin{lemma}
    Let \(0 < \alpha < \beta < 1\). Then \(\comp{T_\alpha} = \comp{T_\beta}\).
\end{lemma}
\begin{proof}
    Assume that \(\comp{T_\alpha} \neq \comp{T_\beta}\). We can assume without loss of generality that \(\comp{T_\alpha} > \comp{T_\beta}\), as in the other case we can take \(T_\alpha(B, A)\) instead of \(T_\alpha(A, B)\) and get the same morph, but parametrised in reverse. We can also assume that for fixed \(\alpha\), \(\beta\) is the smallest value such that \(\comp{T_\alpha} > \comp{T_\beta}\). In this case, there are two regions \(P\) and \(Q\) of \(\partition(A, B)\) or \(\partition(B, A)\) that are disjoint and in different components of \(T_\alpha\), but intersect and are in the same component of \(T_\beta\). This is because, as a consequence of \cref{lem:voronoi-morph-lipschitz}, components cannot appear or disappear. In the following we assume \(P, Q \in \partition(A, B)\); the arguments for when one or both are in \(\partition(B, A)\) are identical.

    As \(P_\beta \cap Q_\beta \neq \emptyset\), there must be some point \(p\) in both \(P_\beta\) and \(Q_\beta\). As both \(P_\beta\) and \(Q_\beta\) are formed by regions moving towards the closest point on the other shape, this point is then on the intersection of two shortest paths between \(A\) and \(B\). Let \(a_1\), \(b_1\), \(a_2\) and \(b_2\) be the endpoints of these paths intersecting in \(p\). One of the two segments \(\seg{pb_1}\), \(\seg{pb_2}\) will be the shortest; assume without loss of generality that it is \(\seg{pb_1}\). In this case the path \(a_2pb_1\) is shorter than \(a_2pb_2\), and by the triangle inequality \(b_1\) must be closer to \(a_2\) than \(b_2\).

    This contradicts the assumption that \(b_2\) was the closest point to \(a_2\). We conclude that such shortest paths can never intersect, and therefore \(P_\alpha \cap Q_\alpha = \emptyset\) for any \(\alpha \in (0, 1)\). As such, components can never merge or split for \(\alpha \in (0, 1)\), and as they also cannot appear or disappear by \cref{lem:voronoi-morph-lipschitz}, the statement in the lemma follows.
\end{proof}

\noindent
Note that the number of components can change at \(\alpha = 0\) or \(\alpha = 1\), as in these limit cases elements of \(\partition_\alpha(A, B)\) and \(\partition_\alpha(B, A)\) turn into points or line segments. Using the strategy from this proof, it also follows that \(\partition_\alpha(A, B)\) and \(\partition_{1 - \alpha}(B, A)\) are interior-disjoint. An interesting corollary of this observation is that the area \(\norm{T_\alpha}\) of \(T_\alpha\) is bounded from below by \((1 - \alpha)^2 \norm{A} + \alpha^2 \norm{B}\), which is attained when both shapes are disjoint and all parts are moving to a finite number of points (vertices) on the other shape.

\subsection{A variant morph}
One problem with the Voronoi morph is that it can introduce many slits into the boundary, thereby greatly increasing the perimeter of the shape. This is because parts of the input that have different closest points on the other shape will tend to move away from each other. We present a variant of the Voronoi morph that tries to reduce these problems. As it uses both the Voronoi morph and the dilation morph, we call this variant the \emph{mixed morph}. The mixed morph \(M_{\alpha,\phi}\) is defined as follows:
\begin{equation*}
    M_{\alpha,\phi}(A, B) \coloneqq ((T_\alpha(A, B) \oplus D_\phi) \ominus D_\phi) \cap S_\alpha \text{,}
\end{equation*}
where \(\ominus\) is the Minkowski difference, defined as \(A \ominus B \coloneqq (A^c \oplus B)^c\), where \(A^c\) is the complement of \(A\). Taking a Minkowski sum with a disk is also known as \emph{dilation}, and the Minkowski difference with a disk is known as \emph{erosion}. Performing first a dilation and then an erosion with disks of the same radius is known as \emph{closing}, and can be used to close small gaps and holes in a shape without modifying the rest too much. The closing operator is widely used and studied in the field of image analysis~\cite{haralick1987image}.

The resulting morph may no longer be a Hausdorff morph: we may have increased the Hausdorff distance by closing certain gaps or holes. We therefore intersect the closed version of \(T_\alpha\) with the dilation morph \(S_\alpha\), so that gaps that are necessary to obtain the appropriate Hausdorff distance are maintained. This results in the mixed morph \( M_{\alpha,\phi}\).

The mixed morph has a new parameter, \(\phi\), being the radius of the disk used in the closing. Note that \(M_{\alpha, 0} = T_\alpha\). We can show that \(M_{\alpha,\phi}\) contains all shapes obtained with the same \(\alpha\) but smaller value of \(\phi\):

\begin{lemma}\label{lem:mixed-morph-containment}
    Let \(\phi, \psi \in \mathbb{R}^+\) and \(\phi \leq \psi\). Then \(M_{\alpha,\phi} \subseteq M_{\alpha,\psi}\).
\end{lemma}
\begin{proof}
    Let us assume that \(M_{\alpha,\phi} \supset M_{\alpha,\psi}\) instead. Then there is some point \(p\) such that \(p \in M_{\alpha,\phi}\), but \(p \notin M_{\alpha,\psi}\). There are two reasons why \(p\) might not be in \(M_{\alpha,\psi}\): either \(p \notin T_\alpha \oplus D_\psi\), or \(p \in T_\alpha \oplus D_\psi\) but \(p \notin (T_\alpha \oplus D_\psi) \ominus D_\psi\).

    It can clearly not be the case that \(p \in M_{\alpha,\phi}\) but \(p \notin T_\alpha \oplus D_\psi\): \(M_{\alpha,\phi}\) is a subset of \(T_\alpha \oplus D_\phi\), and as \(\phi \leq \psi\), we have that \(T_\alpha \oplus D_\phi \subseteq T_\alpha \oplus D_\psi\).

    It must then be the case that \(p \in T_\alpha \oplus D_\psi\) but \(p \notin (T_\alpha \oplus D_\psi) \ominus D_\psi\). In this case, the distance between \(p\) and the boundary \(\partial\dilate{T_\alpha}\) of \(T_\alpha \oplus D_\psi\) must be less than \(\psi\). Let \(q \in \partial\dilate{T_\alpha}\) be the point on the boundary closest to \(p\). As \(p \in T_\alpha \oplus D_\phi\) and \(T_\alpha \oplus D_\phi \subseteq T_\alpha \oplus D_\psi\), the segment \(\seg{pq}\) must intersect the boundary of \(T_\alpha \oplus D_\phi\) in some point \(q'\). We must have that \(d(p,q') \geq \phi\), or \(p\) would not be in \(M_{\alpha,\phi}\), and we must have \(d(q,q') \geq \psi - \phi\), as \(T_\alpha \oplus D_\psi = (T_\alpha \oplus D_\phi) \oplus D_{\psi - \phi}\). But then, by the triangle inequality, \(d(p,q) \leq d(p,q') + d(q,q') \geq \psi\), which is a contradiction. Hence, \(p \in M_{\alpha,\psi}\). As this holds for all \(p \in M_{\alpha,\phi}\), the statement in the lemma follows.
\end{proof}

Note that this means we now have the following hierarchical containment of morphs: \(T_\alpha \subseteq M_{\alpha,\phi} \subseteq M_{\alpha,\psi} \subseteq S_\alpha\), for \(\phi \leq \psi\). As \(T_\alpha\) is a Hausdorff morph, and \(S_\alpha\) is the maximal Hausdorff morph, this shows that \(M_{\alpha,\phi}\) is a Hausdorff morph as well. However, \(M_{\alpha,\phi}\) is not 1-Lipschitz continuous: components may suddenly merge when their distance falls below \(2\phi\).

\subsection{Algorithm}
To give an algorithm for computing \(T_\alpha\), we assume \(A\) and \(B\) are (sets of) polygons, possibly with holes. As \(T_\alpha\) is based on moving all points on the one shape to the closest point on the other shape, we can compute the Voronoi diagram of each input shape, and then use these to partition the other shapes. This gives us a partitioning of \(A\) into pieces that overlap \(B\), or have the same closest point or edge on \(B\), and vice versa. 
Pieces of \(A\) completely inside \(B\) are unchanged, pieces with a vertex as closest element are scaled uniformly towards that vertex by a factor \(\alpha\), and pieces with an edge as closest element are scaled perpendicular to the supporting line of that edge by a factor \(\alpha\). For pieces of \(B\) we do the same, except that we scale them with a factor \(1 - \alpha\). \cref{fig:voronoi-morph-example} shows an example of how a shape \(A\) is partitioned by the Voronoi diagram \(V(B)\) of \(B\), and each piece is scaled towards the closest point on \(B\).

Given this algorithm, we can also straightforwardly compute \(M_{\alpha,\phi}\) by computing \(T_\alpha\) and \(S_\alpha\), dilating and eroding \(T_\alpha\) by a distance \(\phi\), and then intersecting the result with \(S_\alpha\). 

Our computations rely solely on Voronoi diagrams, Minkowski sums and differences with disks, intersections and unions of polygons, all of which can be found in standard books or surveys~\cite{agarwal2008state,aurenhammer2013voronoi,de2008} and an intermediate shape can be calculated in \(O(n^2\log n)\) time. 

\section{Experiments}
We compare the dilation, Voronoi and mixed morphs experimentally on three data sets. The first data set is a collection of outlines of animals taken from~\cite{bouts2016mapping}. The second is a selection of the outlines of European countries obtained from the Thematic Mapping World Borders data set;\footnote{\url{http://www.thematicmapping.org/downloads/world_borders.php}} we use the outlines of Austria, Belgium, Croatia, Czechia, France, Germany, Greece, Ireland, Italy, the Netherlands, Poland, Spain and Sweden. For these two sets we compute the morphs for all pairs of animals and all pairs of countries in the sets. None of the three morphs is translation-invariant or scale-invariant, so it matters where we place the shapes with respect to each other and what sizes they initially have.
We choose to scale
the shapes to have the same area and translate them to have a common centroid. 

The third data set is a small collection of words and letters manually traced as polygons. We use three pairs of words (wish/luck, kick/stuff, try/it), and the letters f, i and u in a serif and a sans serif font. Observe that our morphs could in theory be used to define an infinite family of fonts by interpolating between the glyphs of each element. For these experiments we do not scale the shapes but use the font size, and we align them manually.

For each experiment, we measure the area, perimeter, number of components and number of holes of the morph for \(\alpha\) values starting at zero and increasing in steps of \(1/8\). The parameter \(\phi\) of the mixed morph was universally set to \(0.02\) based on preliminary experimentation.

It is not necessarily insightful to compare areas and especially perimeters between experiments. To make the results more comparable, we make the assumption that an ideal morph linearly interpolates the area and perimeter between those of the input shapes. For each experiment, we can then give the ratio between the measured area and perimeter and these ``ideal'' values. For the number of components and holes this is less meaningful, as these are discrete values, so we simply record the numbers directly.

Each morphing method was implemented in C++, using Boost\footnote{\url{https://www.boost.org}} to calculate intersections and unions of polygons, Voronoi diagrams, and Minkowski sums. Although efficiency is not the focus of this paper, running all our experiments only took a few minutes in total.

\begin{table}[p]
    \centering
    \caption{The distributions of areas for each morphing method over all experiments for all nine tested values of \(\alpha\), separated by experiment category.}
    \label{tab:summary-area}
    \begin{tabular}{l*{6}{c}}
        \toprule
        & \multicolumn{2}{c}{Dilation} & \multicolumn{2}{c}{Voronoi} & \multicolumn{2}{c}{Mixed}\\
        \cmidrule(lr){2-3}
        \cmidrule(lr){4-5}
        \cmidrule(lr){6-7}
        Category & Mean & Std. Dev. & Mean & Std. Dev. & Mean & Std. Dev.\\
        \midrule
        Animals & 1.977 & 0.763 & 0.969 & 0.024 & 0.986 & 0.019\\
        Countries & 2.249 & 1.498 & 0.960 & 0.039 & 0.987 & 0.039\\
        Text & 2.118 & 1.046 & 0.980 & 0.035 & 0.989 & 0.028\\
        \bottomrule\\
    \end{tabular}
\end{table}

\begin{table}[p]
    \centering
    \caption{The distributions of perimeters for each morphing method over all experiments for all nine tested values of \(\alpha\), separated by experiment category.}
    \label{tab:summary-perimeter}
    \begin{tabular}{l*{6}{c}}
        \toprule
        & \multicolumn{2}{c}{Dilation} & \multicolumn{2}{c}{Voronoi} & \multicolumn{2}{c}{Mixed}\\
        \cmidrule(lr){2-3}
        \cmidrule(lr){4-5}
        \cmidrule(lr){6-7}
        Category & Mean & Std. Dev. & Mean & Std. Dev. & Mean & Std. Dev.\\
        \midrule
        Animals & 0.857 & 0.137 & 1.725 & 0.432 & 1.183 & 0.155\\
        Countries & 0.876 & 0.237 & 1.610 & 0.471 & 1.129 & 0.184\\
        Text & 0.955 & 0.142 & 1.401 & 0.418 & 1.155 & 0.192\\
        \bottomrule\\
    \end{tabular}
\end{table}

\begin{table}
    \centering
    \caption{The distributions of the number of components and holes for each morphing method for all tested values of \(\alpha\) except \(0\) and \(1\). This only includes the animals data set, as these shapes have only one component and no holes.}
    \label{tab:summary-topology}
    \begin{tabular}{l*{6}{c}}
        \toprule
        & \multicolumn{2}{c}{Dilation} & \multicolumn{2}{c}{Voronoi} & \multicolumn{2}{c}{Mixed}\\
        \cmidrule(lr){2-3}
        \cmidrule(lr){4-5}
        \cmidrule(lr){6-7}
        Category & Mean & Std. Dev. & Mean & Std. Dev. & Mean & Std. Dev.\\
        \midrule
        Components & 1.004 & 0.063 & 18.556 & 8.089 & 5.317 & 3.213\\
        Holes & 0.218 & 0.602 & 2.544 & 2.699 & 0.218 & 0.532\\
        \bottomrule\\
    \end{tabular}
\end{table}

\begin{figure}[p]
    \centering
    \begin{tikzpicture}
        \begin{axis}[
            title = {Animals},
            cycle list name = customcycle,
            ymin = 0.6, ymax = 3.7,
            width = 0.5\textwidth,
            height = 0.4\textwidth,
            xlabel = {\(\alpha\)},
            ylabel = {area},
            xtick distance = 0.25,
            ytick distance = 0.5,
            minor x tick num = 1,
            legend entries = {Dilation, Voronoi, Mixed},
            legend columns = 3,
            legend to name = morphtypelegend
        ]

            \addplot+ table [x=Alpha,y=Dilation] {data/animals-aggregated-alpha-area.txt};
            \addplot+ table [x=Alpha,y=Voronoi] {data/animals-aggregated-alpha-area.txt};
            \addplot+ table [x=Alpha,y=Mixed] {data/animals-aggregated-alpha-area.txt};
        \end{axis}
    \end{tikzpicture}%
    \hfill
    \begin{tikzpicture}
        \begin{axis}[
            title = {Countries},
            cycle list name = customcycle,
            ymin = 0.6, ymax = 3.7,
            width = 0.5\textwidth,
            height = 0.4\textwidth,
            xlabel = {\(\alpha\)},
            ylabel = {area},
            xtick distance = 0.25,
            ytick distance = 0.5,
            minor x tick num = 1,
        ]

            \addplot+ table [x=Alpha,y=Dilation] {data/countries-aggregated-alpha-area.txt};
            \addplot+ table [x=Alpha,y=Voronoi] {data/countries-aggregated-alpha-area.txt};
            \addplot+ table [x=Alpha,y=Mixed] {data/countries-aggregated-alpha-area.txt};
        \end{axis}
    \end{tikzpicture}
    \ref{morphtypelegend}
    \caption{The average area over all experiments as a function of \(\alpha\), for both the animals and countries data sets.}
    \label{fig:animals-countries-average-areas}
\end{figure}
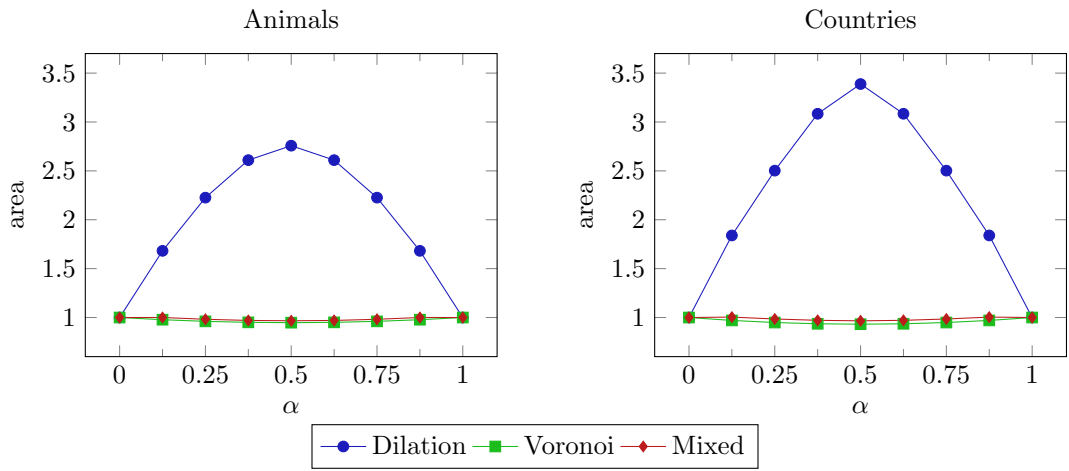

\begin{figure}
    \centering
    \begin{tikzpicture}
        \begin{axis}[
            title = {Animals},
            cycle list name = customcycle,
            ymin = 0.6, ymax = 2.1,
            width = 0.5\textwidth,
            height = 0.4\textwidth,
            xlabel = {\(\alpha\)},
            ylabel = {perimeter},
            xtick distance = 0.25,
            ytick = {0.7,1.0,...,2.0},
            minor x tick num = 1,
        ]

            \addplot+ table [x=Alpha,y=Dilation] {data/animals-aggregated-alpha-perimeter.txt};
            \addplot+ table [x=Alpha,y=Voronoi] {data/animals-aggregated-alpha-perimeter.txt};
            \addplot+ table [x=Alpha,y=Mixed] {data/animals-aggregated-alpha-perimeter.txt};
        \end{axis}
    \end{tikzpicture}%
    \hfill
    \begin{tikzpicture}
        \begin{axis}[
            title = {Countries},
            cycle list name = customcycle,
            ymin = 0.6, ymax = 2.1,
            width = 0.5\textwidth,
            height = 0.4\textwidth,
            xlabel = {\(\alpha\)},
            ylabel = {perimeter},
            xtick distance = 0.25,
            ytick = {0.7,1.0,...,2.0},
            minor x tick num = 1,
        ]

            \addplot+ table [x=Alpha,y=Dilation] {data/countries-aggregated-alpha-perimeter.txt};
            \addplot+ table [x=Alpha,y=Voronoi] {data/countries-aggregated-alpha-perimeter.txt};
            \addplot+ table [x=Alpha,y=Mixed] {data/countries-aggregated-alpha-perimeter.txt};
        \end{axis}
    \end{tikzpicture}
    \ref{morphtypelegend}
    \caption{The average perimeter over all experiments as a function of \(\alpha\), for both the animals and countries data sets.}
    \label{fig:animals-countries-average-perimeters}
\end{figure}
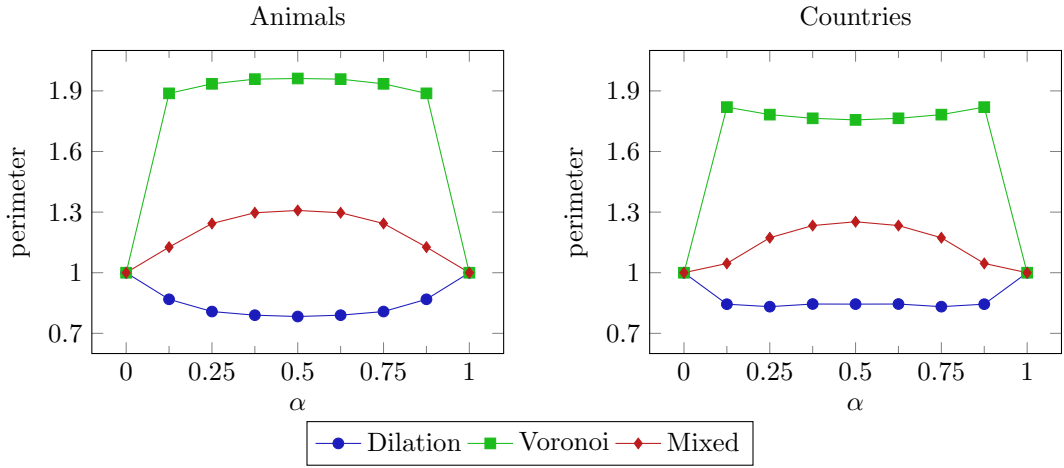

\section{Results}
A summary of our measurements of area and perimeter can be seen in \cref{tab:summary-area,tab:summary-perimeter}. A summary of the number of components and holes for only the animals data set can be seen in \cref{tab:summary-topology}; we exclude the other data sets because the inputs have different numbers of components. Topological measurements for all experiments can be viewed in \cref{tab:topology} in \cref{sec:topology-tables}. We note that the Voronoi and mixed morphs sometimes have spurious holes caused by numerical precision issues (e.g., the Voronoi morph should not have an intermediate shape with five holes in our experiment with the letter i). Animations of the different morphs for each experiment can be viewed online.\footnote{\url{https://hausdorff-morphing.github.io}}

In \cref{fig:animals-countries-average-areas} we can see that the average area of the dilation morph quickly grows as \(\alpha\) increases, until reaching its peak at \(\alpha = 1/2\), to about three times the desired size. For the perimeter we see the opposite trend, with the dilation morph typically having a lower perimeter than desired. This is a consequence of the dilation erasing details in the boundary of the input shapes. We can see in \cref{fig:animals-countries-average-perimeters} that this happens more quickly in the experiments with country shapes. This is expected, as most of the country shapes have more sharp coastline features and islands that quickly disappear, whereas the animal shapes are generally smoother and only have one component.

\begin{figure}
    \centering
    \includegraphics[width=0.32\textwidth]{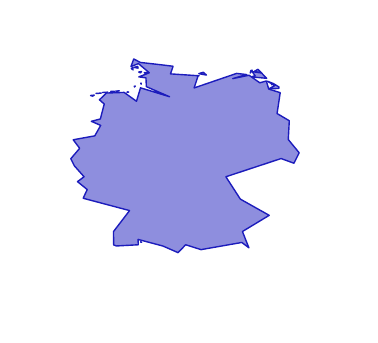}
    \hfill
    \includegraphics[width=0.32\textwidth]{figures/germany.pdf}
    \hfill
    \includegraphics[width=0.32\textwidth]{figures/germany.pdf}

    \includegraphics[width=0.32\textwidth]{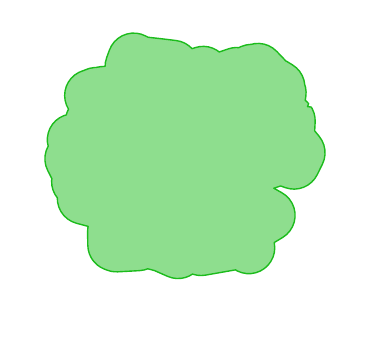}
    \hfill
    \includegraphics[width=0.32\textwidth]{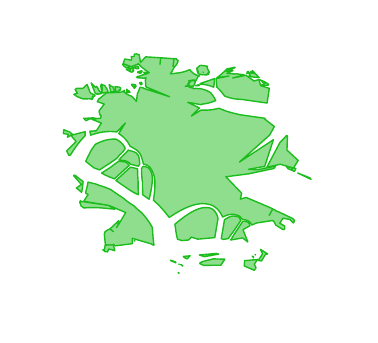}
    \hfill
    \includegraphics[width=0.32\textwidth]{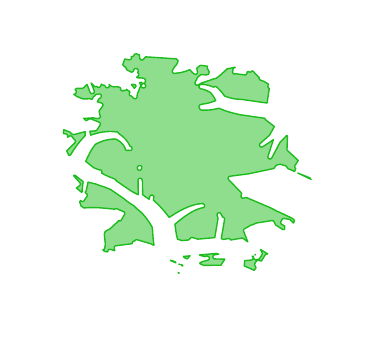}

    \includegraphics[width=0.32\textwidth]{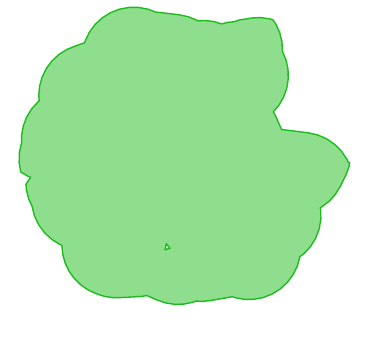}
    \hfill
    \includegraphics[width=0.32\textwidth]{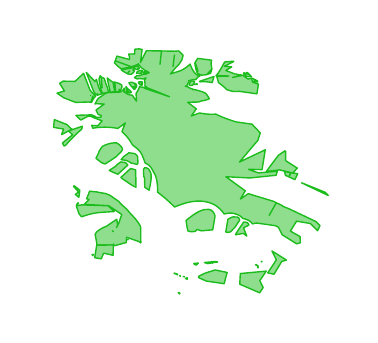}
    \hfill
    \includegraphics[width=0.32\textwidth]{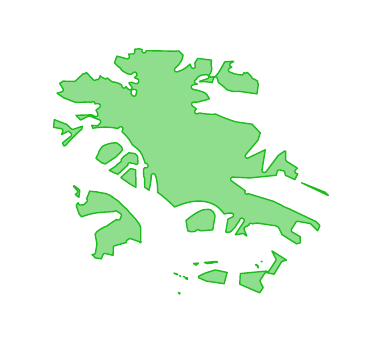}

    \includegraphics[width=0.32\textwidth]{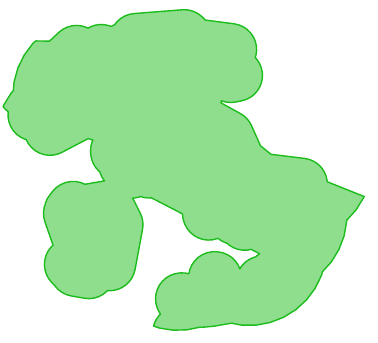}
    \hfill
    \includegraphics[width=0.32\textwidth]{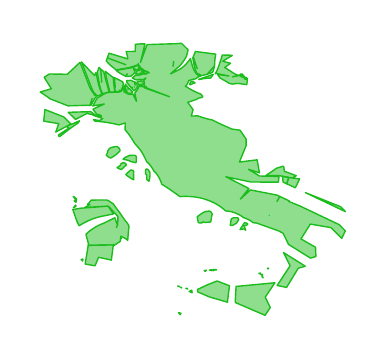}
    \hfill
    \includegraphics[width=0.32\textwidth]{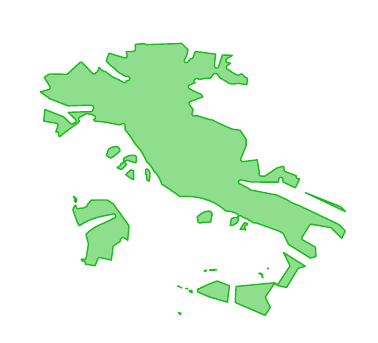}

    \includegraphics[width=0.32\textwidth]{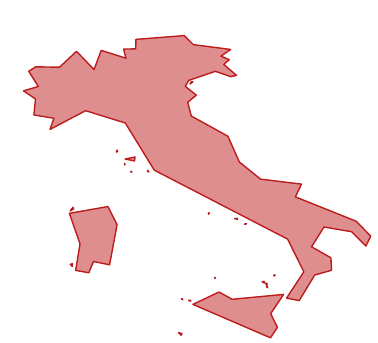}
    \hfill
    \includegraphics[width=0.32\textwidth]{figures/italy.pdf}
    \hfill
    \includegraphics[width=0.32\textwidth]{figures/italy.pdf}
    \caption{Intermediate shapes for \(\alpha \in \{0, 1/4, 1/2, 3/4, 1\}\) when morphing between the outlines of Germany and Italy. The columns show the dilation morph, Voronoi morph and mixed morph from left to right.}
    \label{fig:germany-italy}
\end{figure}

Our Voronoi morph on average has an area that is much closer to the desired value, and with much lower variance than the dilation morph. However, we see that on average the perimeter is much higher than the desired value. This is because points on opposite sides of a Voronoi edge move in different directions, causing new boundaries to appear in the interior of a shape as soon as \(\alpha > 0\). We can see this happen in the middle column of \cref{fig:germany-italy}, and this is reflected in \cref{fig:animals-countries-average-perimeters}, where we see the perimeter sharply increase and then stay mostly the same, before sharply dropping back down.

Our mixed morph achieves its purpose of reducing the perimeter of the Voronoi morph: the measured perimeters are close to the desired values, while the measured areas stay comparable to those of the Voronoi morph. In \cref{fig:animals-countries-average-perimeters}, we see that the perimeter typically still increases during the morphing process, but does not jump up sharply as soon as \(\alpha > 0\). This is because the small value of \(\phi\) lets us close only the narrow gaps that appear around the edges of the Voronoi diagram, but not the gaps that develop as pieces of the shapes move apart significantly. We can see this when comparing the middle and right columns of \cref{fig:germany-italy}: fewer gaps are closed at \(\alpha = 1/2\) than at the other time values.

In addition to area and perimeter, we also tracked the number of components and holes for each morph type. We observe that for the dilation morph, there is an intermediate shape with only one component in all but one of our experiments (see \cref{tab:topology} in \cref{sec:topology-tables}), showing that this morph tends to turn everything into a blob during the morphing process. On the other hand, the Voronoi morph tends to have an intermediate shape with a number of components much larger than either of the input shapes. The mixed morph exhibits neither of these behaviours. This is illustrated in \cref{fig:germany-italy}.

Inspecting the morphs visually (Figures~\ref{fig:germany-italy}--\ref{fig:try-it}, partly in \cref{sec:more-examples}), our mixed morph looks quite reasonable, especially when the area of symmetric difference between the input shapes is small. In many cases, the intermediate shape at \(\alpha = 1/2\) is a recognisable mix of the two input shapes. This is not the case for the dilation morph, where the Hausdorff distance needs to be very small compared to the size of the input shapes for it to look good. 
For instance, when one shape has some small islands far away, the dilation morph will grow to have a very large area, whereas with the Voronoi and mixed morphs, the islands just slowly move towards the closest point on the other shape; see \cref{fig:france-spain} in \cref{sec:more-examples}. 
However, both the Voronoi and mixed morph can still look bad when the area of symmetric difference is large. It may therefore be best to align the input shapes such that the area of symmetric difference is minimised, rather than simply aligning the centroids.

The morphs generally look less convincing on our experiments with text, as the shapes can be very different. For single letters (\cref{fig:i} in \cref{sec:more-examples}) the morphs can look convincing, but when morphing between words, especially of different numbers of letters, the intermediate shape at \(\alpha = 1/2\) does not necessarily resemble both input shapes (\cref{fig:try-it} in \cref{sec:more-examples}). However, the intermediate shapes at \(\alpha = 1/4\) and \(\alpha = 3/4\) still do clearly resemble input shapes \(A\) and \(B\), respectively, for the Voronoi and mixed morphs, but less so for the dilation morph. A better approach to morphing text may be to morph on a per-letter basis, rather than treating the whole text as a single shape. Some strategy would then have to be devised that determines which letter will morph to which, and how to deal with different Hausdorff distances between the letter pairs.

Both the Voronoi morph and the mixed morph often have small parts separating, moving, and then merging somewhere else (for example, the beak in the bird-to-ostrich morphs on \url{https://hausdorff-morphing.github.io}). Such artifacts may be circumvented by choosing a slightly warped Voronoi diagram, but this upsets the simplicity of the current methods.
We can sometimes notice in the animations that the mixed morph is indeed not Lipschitz continuous, but since $\phi$ is rather small, this does not show clearly.

\section{Conclusion}
We introduced a new abstract morphing method based on Voronoi diagrams. This new method satisfies the same bounds on the Hausdorff distance as the previously introduced dilation morph, and is also 1-Lipschitz continuous. We have shown experimentally that the intermediate shapes of the Voronoi morph have areas that more closely match those of the input shapes than the dilation morph, but tends to have a perimeter that is larger than desired. To remedy this, we introduced a variant morph, the mixed morph, that we experimentally show to reduce this problem of increasing the perimeter. This mixed morph still satisfies the bounds on the Hausdorff distance, but is no longer 1-Lipschitz continuous. Our experimental analysis is the first we are aware of that analyses
the development of area, perimeter, number of components and number of holes throughout the morphs.

An interesting open question is whether we can prevent the increase in perimeter caused by the Voronoi morph without losing 1-Lipschitz continuity. This would require somehow anticipating the moment when two pieces of boundary will meet, and smoothly bridging the gap between them over time, instead of just instantly filling it.
To optimise the mixed morph, we can study the effects of choosing different $\phi$, or even changing $\phi$ throughout the morph.
Another direction is to develop other morphs that guarantee a smooth change of some distance measure other than the Hausdorff distance; we noted that it is unclear how to employ the Fr\'echet distance for morphing in the presence of multiple components.

A more practically oriented direction for further research would be to develop a less naive method of filling gaps than the mixed morph. It does not necessarily make sense to use the same radius for the closing operator everywhere, which sometimes closes gaps that will be opened again.
However, any adaptation of this type will disrupt the conceptual simplicity of the Voronoi and mixed morphs.

\bibliography{bibliography}

\clearpage
\appendix

\section{More example morphs}\label{sec:more-examples}
\begin{minipage}{\textwidth}
    \centering
    \captionsetup{type=figure}
    \includegraphics[width=0.32\textwidth]{figures/butterfly.pdf}
    \hfill
    \includegraphics[width=0.32\textwidth]{figures/butterfly.pdf}
    \hfill
    \includegraphics[width=0.32\textwidth]{figures/butterfly.pdf}

    \includegraphics[width=0.32\textwidth]{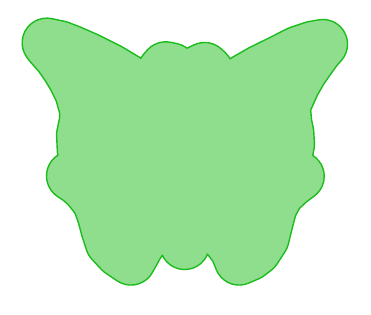}
    \hfill
    \includegraphics[width=0.32\textwidth]{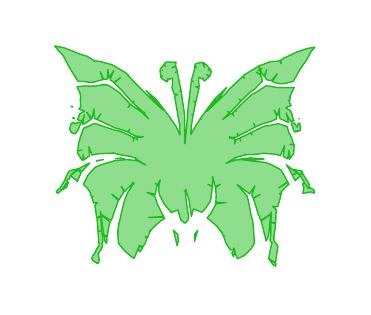}
    \hfill
    \includegraphics[width=0.32\textwidth]{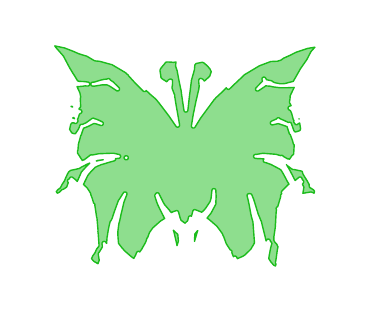}

    \includegraphics[width=0.32\textwidth]{figures/butterfly-spider-dilation-0.5.pdf}
    \hfill
    \includegraphics[width=0.32\textwidth]{figures/butterfly-spider-voronoi-0.5.pdf}
    \hfill
    \includegraphics[width=0.32\textwidth]{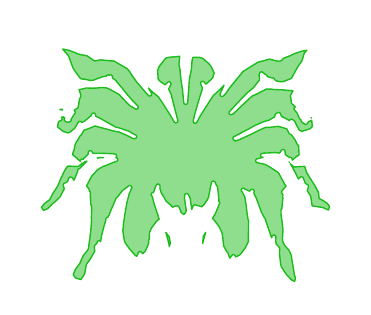}

    \includegraphics[width=0.32\textwidth]{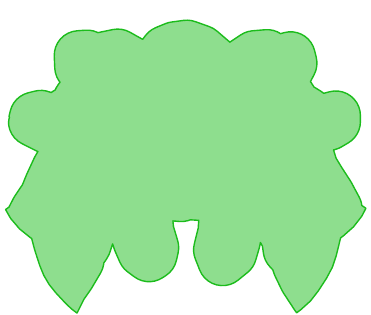}
    \hfill
    \includegraphics[width=0.32\textwidth]{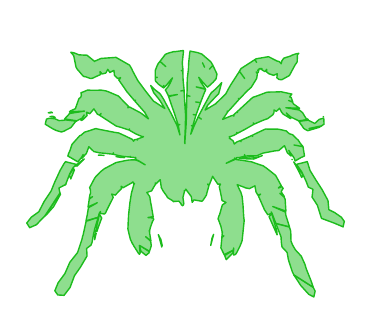}
    \hfill
    \includegraphics[width=0.32\textwidth]{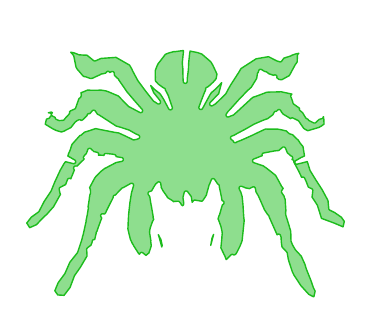}

    \includegraphics[width=0.32\textwidth]{figures/spider.pdf}
    \hfill
    \includegraphics[width=0.32\textwidth]{figures/spider.pdf}
    \hfill
    \includegraphics[width=0.32\textwidth]{figures/spider.pdf}
    \captionof{figure}{Intermediate shapes for \(\alpha \in \{0, 1/4, 1/2, 3/4, 1\}\) when morphing between the outlines of a butterfly and a spider. The columns show the dilation morph, Voronoi morph and mixed morph from left to right.}
    \label{fig:butterfly-spider}
\end{minipage}

\begin{figure}
    \centering
    \includegraphics[width=0.32\textwidth]{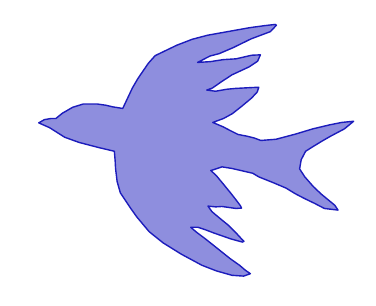}
    \hfill
    \includegraphics[width=0.32\textwidth]{figures/bird.pdf}
    \hfill
    \includegraphics[width=0.32\textwidth]{figures/bird.pdf}

    \includegraphics[width=0.32\textwidth]{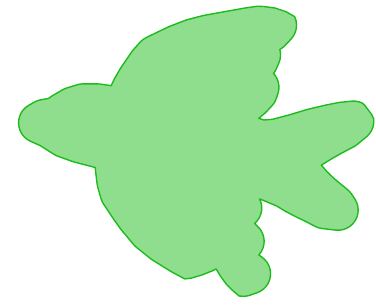}
    \hfill
    \includegraphics[width=0.32\textwidth]{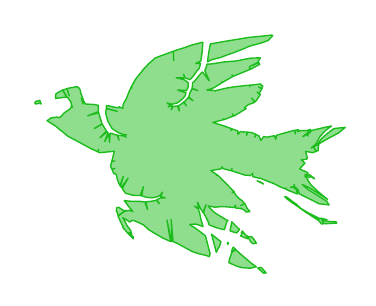}
    \hfill
    \includegraphics[width=0.32\textwidth]{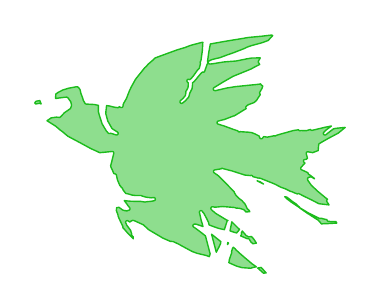}

    \includegraphics[width=0.32\textwidth]{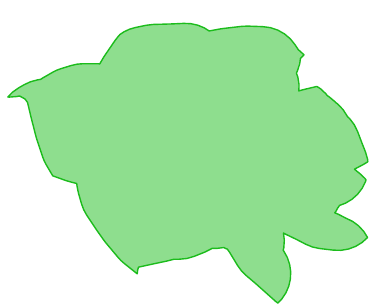}
    \hfill
    \includegraphics[width=0.32\textwidth]{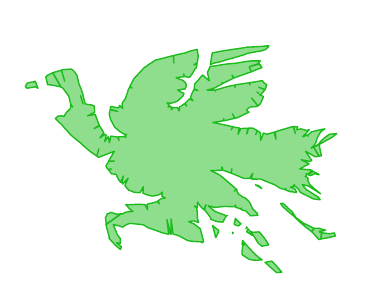}
    \hfill
    \includegraphics[width=0.32\textwidth]{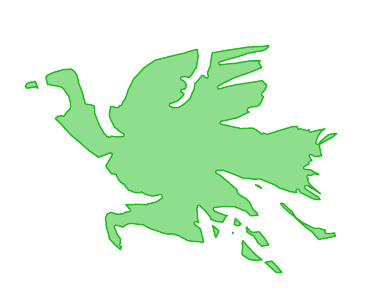}

    \includegraphics[width=0.32\textwidth]{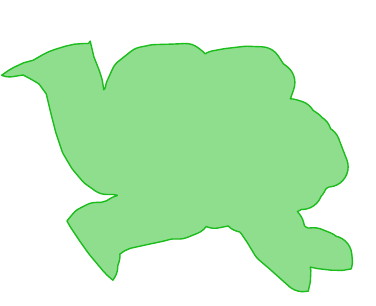}
    \hfill
    \includegraphics[width=0.32\textwidth]{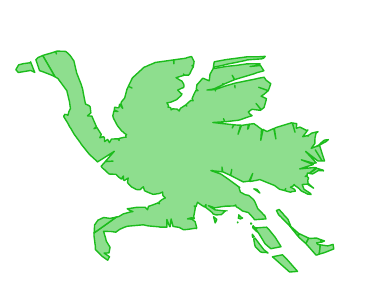}
    \hfill
    \includegraphics[width=0.32\textwidth]{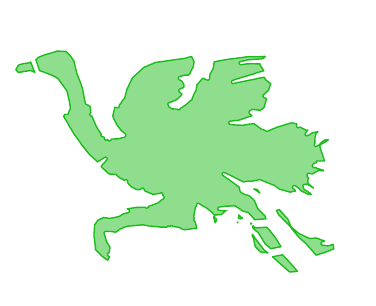}

    \includegraphics[width=0.32\textwidth]{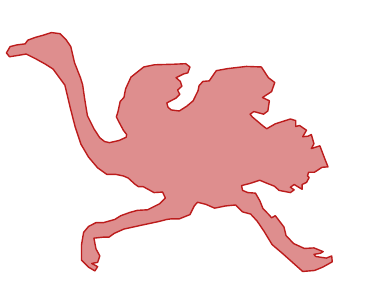}
    \hfill
    \includegraphics[width=0.32\textwidth]{figures/ostrich.pdf}
    \hfill
    \includegraphics[width=0.32\textwidth]{figures/ostrich.pdf}
    \caption{Intermediate shapes for \(\alpha \in \{0, 1/4, 1/2, 3/4, 1\}\) when morphing between the outlines of a bird and an ostrich. The columns show the dilation morph, Voronoi morph and mixed morph from left to right.}
    \label{fig:bird-ostrich}
\end{figure}

\begin{figure}
    \centering
    \includegraphics[width=0.32\textwidth]{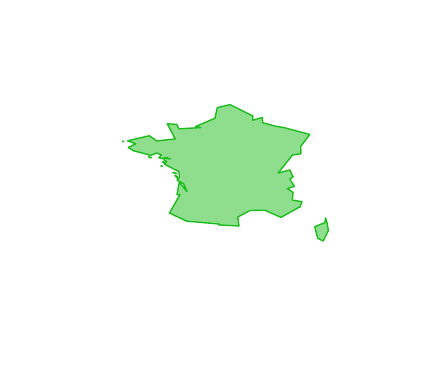}
    \hfill
    \includegraphics[width=0.32\textwidth]{figures/france.pdf}
    \hfill
    \includegraphics[width=0.32\textwidth]{figures/france.pdf}

    \includegraphics[width=0.32\textwidth]{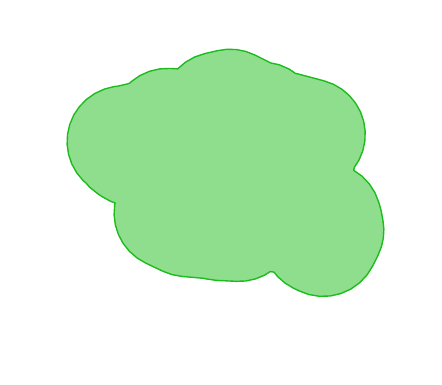}
    \hfill
    \includegraphics[width=0.32\textwidth]{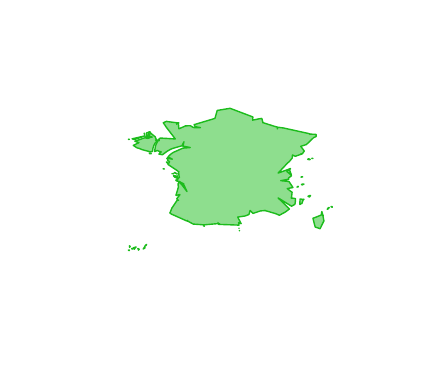}
    \hfill
    \includegraphics[width=0.32\textwidth]{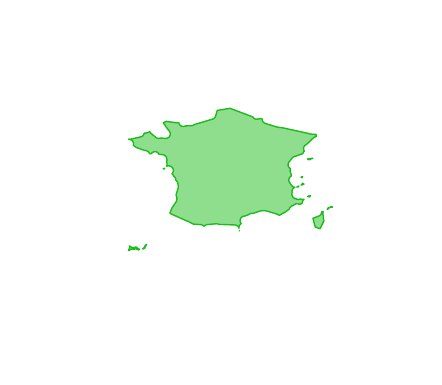}

    \includegraphics[width=0.32\textwidth]{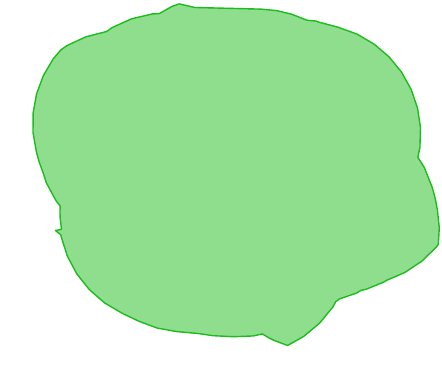}
    \hfill
    \includegraphics[width=0.32\textwidth]{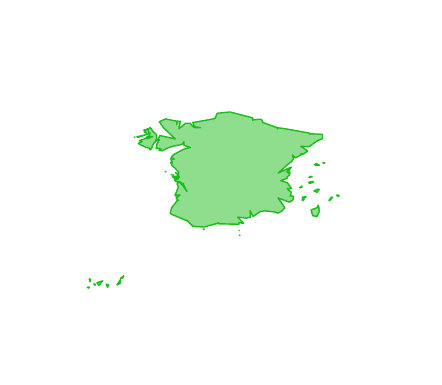}
    \hfill
    \includegraphics[width=0.32\textwidth]{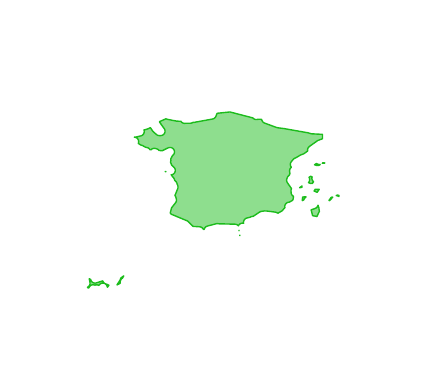}

    \includegraphics[width=0.32\textwidth]{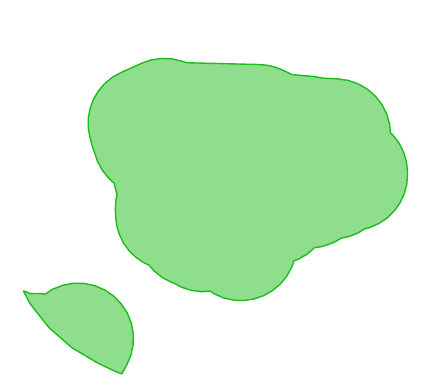}
    \hfill
    \includegraphics[width=0.32\textwidth]{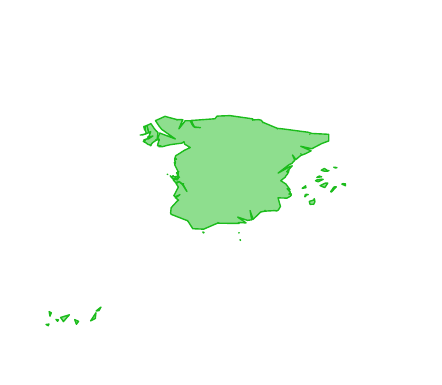}
    \hfill
    \includegraphics[width=0.32\textwidth]{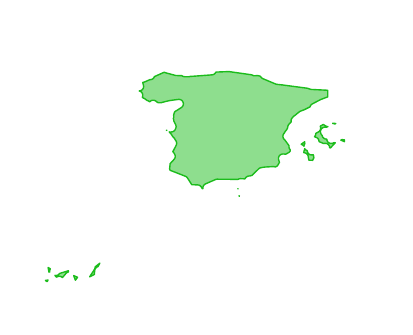}

    \includegraphics[width=0.32\textwidth]{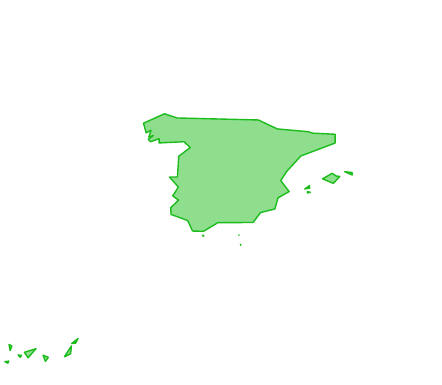}
    \hfill
    \includegraphics[width=0.32\textwidth]{figures/spain.pdf}
    \hfill
    \includegraphics[width=0.32\textwidth]{figures/spain.pdf}
    \caption{Intermediate shapes for \(\alpha \in \{0, 1/4, 1/2, 3/4, 1\}\) when morphing between the outlines of France and Spain. The columns show the dilation morph, Voronoi morph and mixed morph from left to right.}
    \label{fig:france-spain}
\end{figure}

\begin{figure}
    \centering
    \hspace{1cm}
    \includegraphics[width=0.07\textwidth]{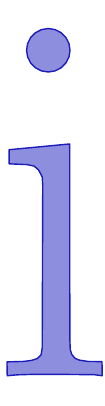}
    \hfill
    \includegraphics[width=0.07\textwidth]{figures/i-serif.pdf}
    \hfill
    \includegraphics[width=0.07\textwidth]{figures/i-serif.pdf}
    \hspace{1cm}

    \hspace{1cm}
    \includegraphics[width=0.07\textwidth]{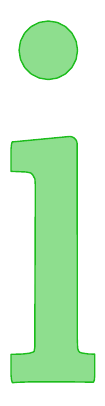}
    \hfill
    \includegraphics[width=0.07\textwidth]{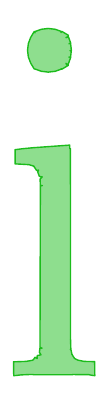}
    \hfill
    \includegraphics[width=0.07\textwidth]{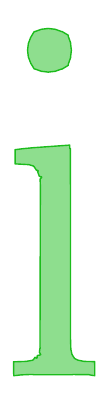}
    \hspace{1cm}

    \hspace{1cm}
    \includegraphics[width=0.07\textwidth]{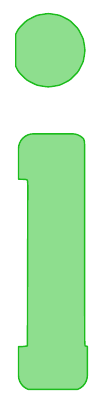}
    \hfill
    \includegraphics[width=0.07\textwidth]{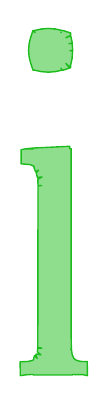}
    \hfill
    \includegraphics[width=0.07\textwidth]{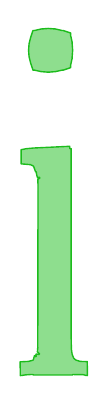}
    \hspace{1cm}

    \hspace{1cm}
    \includegraphics[width=0.07\textwidth]{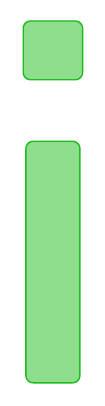}
    \hfill
    \includegraphics[width=0.07\textwidth]{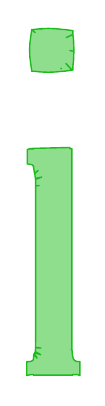}
    \hfill
    \includegraphics[width=0.07\textwidth]{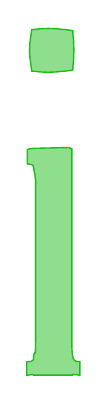}
    \hspace{1cm}

    \hspace{1cm}
    \includegraphics[width=0.07\textwidth]{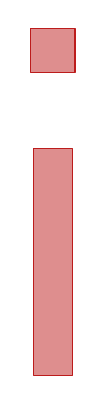}
    \hfill
    \includegraphics[width=0.07\textwidth]{figures/i-sans.pdf}
    \hfill
    \includegraphics[width=0.07\textwidth]{figures/i-sans.pdf}
    \hspace{1cm}
    \caption{Intermediate shapes for \(\alpha \in \{0, 1/4, 1/2, 3/4, 1\}\) when morphing between the outlines of the letter i in two different fonts. The columns show the dilation morph, Voronoi morph and mixed morph from left to right. Note that some artefacts in the Voronoi and mixed morphs, such as on the i's dot, are caused by having polygonal input instead of smooth curves.}
    \label{fig:i}
\end{figure}

\begin{figure}
    \centering
    \includegraphics[width=0.3\textwidth]{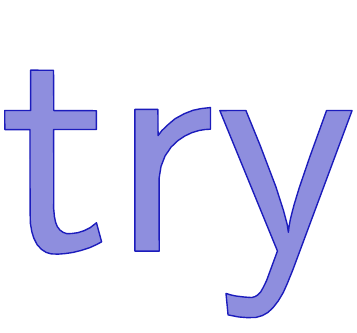}
    \hfill
    \includegraphics[width=0.3\textwidth]{figures/try.pdf}
    \hfill
    \includegraphics[width=0.3\textwidth]{figures/try.pdf}

    \includegraphics[width=0.3\textwidth]{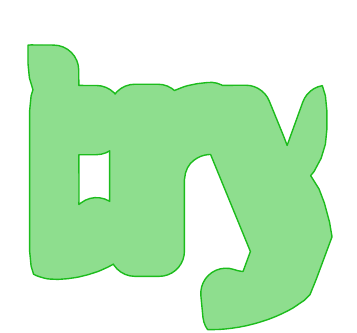}
    \hfill
    \includegraphics[width=0.3\textwidth]{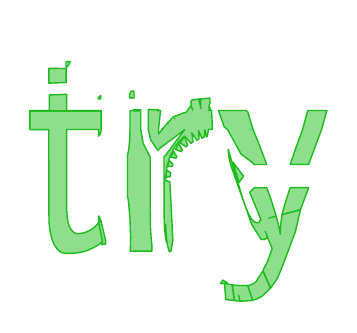}
    \hfill
    \includegraphics[width=0.3\textwidth]{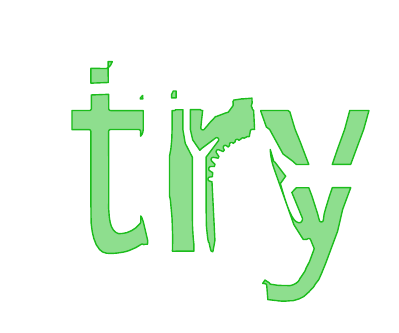}

    \includegraphics[width=0.3\textwidth]{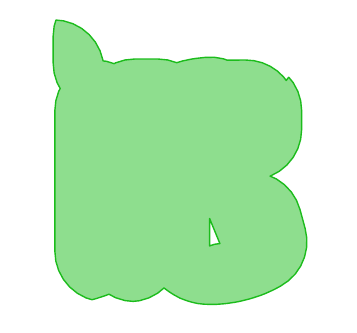}
    \hfill
    \includegraphics[width=0.3\textwidth]{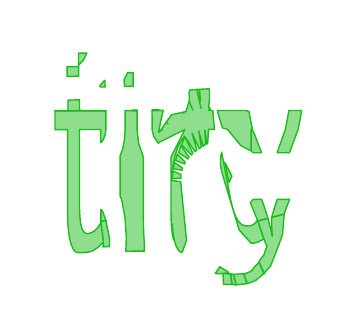}
    \hfill
    \includegraphics[width=0.3\textwidth]{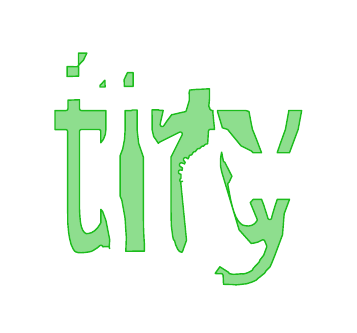}

    \includegraphics[width=0.3\textwidth]{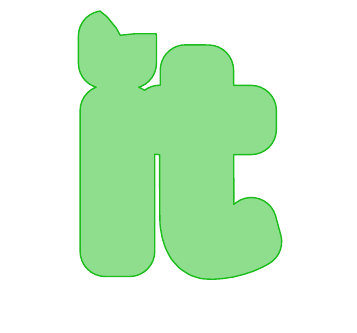}
    \hfill
    \includegraphics[width=0.3\textwidth]{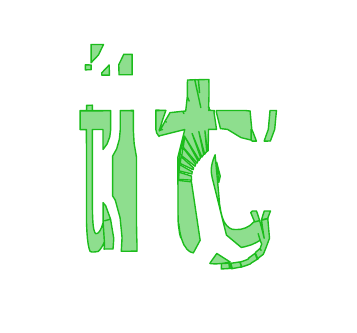}
    \hfill
    \includegraphics[width=0.3\textwidth]{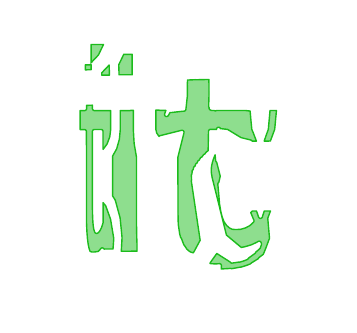}

    \includegraphics[width=0.3\textwidth]{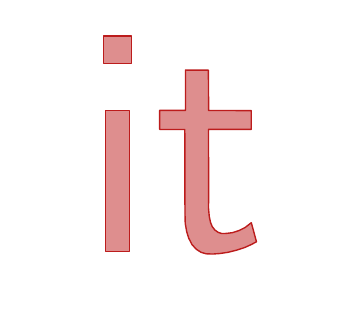}
    \hfill
    \includegraphics[width=0.3\textwidth]{figures/it.pdf}
    \hfill
    \includegraphics[width=0.3\textwidth]{figures/it.pdf}
    \caption{Intermediate shapes for \(\alpha \in \{0, 1/4, 1/2, 3/4, 1\}\) when morphing between the outlines of the words try and it. The columns show the dilation morph, Voronoi morph and mixed morph from left to right. Note that some artefacts in the Voronoi and mixed morphs, such as in the curved part of the letter r, are caused by having polygonal input instead of smooth curves.}
    \label{fig:try-it}
\end{figure}

\clearpage

\section{Topology tables}\label{sec:topology-tables}
\begin{longtable}{l*{10}{c}}
    \caption{The minimum and maximum number of components and the maximum number of holes for each experiment, separated by morph type.}\label{tab:topology}\\
    \toprule
    & \multicolumn{3}{c}{Dilation} & \multicolumn{3}{c}{Voronoi} & \multicolumn{3}{c}{Mixed}\\
    \cmidrule(lr){2-4}
    \cmidrule(lr){5-7}
    \cmidrule(lr){8-10}
    Experiment & min & max & holes & min & max & holes & min & max & holes\\
    \midrule
    \endfirsthead
    \caption[]{(continued from last page)}\\
    \toprule
    & \multicolumn{3}{c}{Dilation} & \multicolumn{3}{c}{Voronoi} & \multicolumn{3}{c}{Mixed}\\
    \cmidrule(lr){2-4}
    \cmidrule(lr){5-7}
    \cmidrule(lr){8-10}
    Experiment & min & max & holes & min & max & holes & min & max & holes\\
    \midrule
    \endhead
    \bottomrule \\
    \endfoot
    bird \(\rightarrow\) butterfly & 1 & 1 & 0 & 1 & 11 & 6 & 1 & 4 & 1\\
    bird \(\rightarrow\) cat & 1 & 1 & 1 & 1 & 12 & 6 & 1 & 5 & 0\\
    bird \(\rightarrow\) dog & 1 & 1 & 1 & 1 & 23 & 4 & 1 & 7 & 1\\
    bird \(\rightarrow\) horse & 1 & 2 & 1 & 1 & 27 & 6 & 1 & 8 & 2\\
    bird \(\rightarrow\) ostrich & 1 & 1 & 0 & 1 & 26 & 5 & 1 & 12 & 0\\
    bird \(\rightarrow\) shark & 1 & 1 & 0 & 1 & 14 & 4 & 1 & 6 & 0\\
    bird \(\rightarrow\) spider & 1 & 1 & 0 & 1 & 30 & 8 & 1 & 9 & 1\\
    bird \(\rightarrow\) turtle & 1 & 1 & 0 & 1 & 21 & 3 & 1 & 9 & 1\\
    butterfly \(\rightarrow\) cat & 1 & 1 & 1 & 1 & 6 & 2 & 1 & 3 & 1\\
    butterfly \(\rightarrow\) dog & 1 & 1 & 1 & 1 & 9 & 4 & 1 & 4 & 1\\
    butterfly \(\rightarrow\) horse & 1 & 1 & 1 & 1 & 28 & 3 & 1 & 9 & 2\\
    butterfly \(\rightarrow\) ostrich & 1 & 1 & 1 & 1 & 11 & 7 & 1 & 4 & 2\\
    butterfly \(\rightarrow\) shark & 1 & 1 & 0 & 1 & 8 & 4 & 1 & 3 & 1\\
    butterfly \(\rightarrow\) spider & 1 & 1 & 1 & 1 & 30 & 9 & 1 & 10 & 2\\
    butterfly \(\rightarrow\) turtle & 1 & 1 & 0 & 1 & 13 & 6 & 1 & 3 & 2\\
    cat \(\rightarrow\) dog & 1 & 1 & 1 & 1 & 17 & 1 & 1 & 3 & 1\\
    cat \(\rightarrow\) horse & 1 & 1 & 1 & 1 & 17 & 2 & 1 & 5 & 1\\
    cat \(\rightarrow\) ostrich & 1 & 1 & 0 & 1 & 13 & 4 & 1 & 4 & 0\\
    cat \(\rightarrow\) shark & 1 & 1 & 0 & 1 & 9 & 0 & 1 & 5 & 0\\
    cat \(\rightarrow\) spider & 1 & 1 & 2 & 1 & 29 & 3 & 1 & 7 & 3\\
    cat \(\rightarrow\) turtle & 1 & 1 & 0 & 1 & 14 & 1 & 1 & 7 & 0\\
    dog \(\rightarrow\) horse & 1 & 1 & 2 & 1 & 31 & 4 & 1 & 9 & 1\\
    dog \(\rightarrow\) ostrich & 1 & 1 & 1 & 1 & 22 & 2 & 1 & 7 & 1\\
    dog \(\rightarrow\) shark & 1 & 1 & 1 & 1 & 17 & 3 & 1 & 6 & 1\\
    dog \(\rightarrow\) spider & 1 & 1 & 3 & 1 & 32 & 2 & 1 & 14 & 2\\
    dog \(\rightarrow\) turtle & 1 & 1 & 1 & 1 & 16 & 1 & 1 & 6 & 0\\
    horse \(\rightarrow\) ostrich & 1 & 1 & 2 & 1 & 27 & 7 & 1 & 15 & 1\\
    horse \(\rightarrow\) shark & 1 & 1 & 1 & 1 & 22 & 5 & 1 & 7 & 1\\
    horse \(\rightarrow\) spider & 1 & 1 & 4 & 1 & 38 & 3 & 1 & 9 & 1\\
    horse \(\rightarrow\) turtle & 1 & 1 & 1 & 1 & 22 & 4 & 1 & 12 & 0\\
    ostrich \(\rightarrow\) shark & 1 & 1 & 0 & 1 & 15 & 8 & 1 & 5 & 0\\
    ostrich \(\rightarrow\) spider & 1 & 1 & 4 & 1 & 38 & 9 & 1 & 17 & 0\\
    ostrich \(\rightarrow\) turtle & 1 & 1 & 0 & 1 & 21 & 12 & 1 & 4 & 0\\
    shark \(\rightarrow\) spider & 1 & 1 & 0 & 1 & 23 & 2 & 1 & 5 & 2\\
    shark \(\rightarrow\) turtle & 1 & 1 & 0 & 1 & 11 & 3 & 1 & 3 & 0\\
    spider \(\rightarrow\) turtle & 1 & 1 & 4 & 1 & 25 & 14 & 1 & 8 & 3\\
    austria \(\rightarrow\) belgium & 1 & 1 & 0 & 1 & 2 & 2 & 1 & 1 & 0\\
    austria \(\rightarrow\) croatia & 1 & 19 & 0 & 1 & 24 & 2 & 1 & 19 & 3\\
    austria \(\rightarrow\) czechia & 1 & 1 & 0 & 1 & 2 & 2 & 1 & 1 & 0\\
    austria \(\rightarrow\) france & 1 & 10 & 0 & 1 & 16 & 2 & 1 & 10 & 1\\
    austria \(\rightarrow\) germany & 1 & 20 & 0 & 1 & 20 & 2 & 1 & 20 & 0\\
    austria \(\rightarrow\) greece & 1 & 68 & 4 & 1 & 79 & 2 & 1 & 68 & 2\\
    austria \(\rightarrow\) ireland & 1 & 4 & 0 & 1 & 12 & 1 & 1 & 4 & 1\\
    austria \(\rightarrow\) italy & 1 & 22 & 1 & 1 & 33 & 1 & 1 & 22 & 1\\
    austria \(\rightarrow\) netherlands & 1 & 9 & 1 & 1 & 15 & 2 & 1 & 9 & 2\\
    austria \(\rightarrow\) poland & 1 & 1 & 0 & 1 & 2 & 2 & 1 & 2 & 0\\
    austria \(\rightarrow\) spain & 1 & 15 & 0 & 1 & 23 & 2 & 1 & 15 & 1\\
    austria \(\rightarrow\) sweden & 1 & 19 & 0 & 1 & 26 & 2 & 1 & 19 & 0\\
    belgium \(\rightarrow\) croatia & 1 & 19 & 1 & 1 & 26 & 3 & 1 & 19 & 4\\
    belgium \(\rightarrow\) czechia & 1 & 1 & 0 & 1 & 1 & 3 & 1 & 1 & 0\\
    belgium \(\rightarrow\) france & 1 & 10 & 0 & 1 & 13 & 1 & 1 & 10 & 2\\
    belgium \(\rightarrow\) germany & 1 & 20 & 0 & 1 & 21 & 2 & 1 & 20 & 2\\
    belgium \(\rightarrow\) greece & 1 & 68 & 6 & 1 & 81 & 1 & 1 & 68 & 4\\
    belgium \(\rightarrow\) ireland & 1 & 4 & 0 & 1 & 9 & 1 & 1 & 4 & 2\\
    belgium \(\rightarrow\) italy & 1 & 22 & 1 & 1 & 30 & 1 & 1 & 22 & 1\\
    belgium \(\rightarrow\) netherlands & 1 & 9 & 1 & 1 & 13 & 3 & 1 & 9 & 2\\
    belgium \(\rightarrow\) poland & 1 & 1 & 0 & 1 & 2 & 2 & 1 & 1 & 0\\
    belgium \(\rightarrow\) spain & 1 & 15 & 0 & 1 & 18 & 2 & 1 & 15 & 0\\
    belgium \(\rightarrow\) sweden & 1 & 19 & 2 & 1 & 25 & 1 & 1 & 19 & 1\\
    croatia \(\rightarrow\) czechia & 1 & 19 & 0 & 1 & 24 & 3 & 1 & 19 & 4\\
    croatia \(\rightarrow\) france & 1 & 19 & 2 & 10 & 47 & 1 & 8 & 22 & 1\\
    croatia \(\rightarrow\) germany & 1 & 20 & 0 & 19 & 55 & 0 & 6 & 20 & 3\\
    croatia \(\rightarrow\) greece & 1 & 68 & 3 & 19 & 121 & 4 & 19 & 68 & 4\\
    croatia \(\rightarrow\) ireland & 1 & 19 & 0 & 4 & 34 & 5 & 3 & 19 & 4\\
    croatia \(\rightarrow\) italy & 1 & 22 & 1 & 19 & 67 & 7 & 19 & 30 & 1\\
    croatia \(\rightarrow\) netherlands & 1 & 19 & 1 & 9 & 39 & 2 & 7 & 19 & 6\\
    croatia \(\rightarrow\) poland & 1 & 19 & 0 & 1 & 34 & 1 & 1 & 19 & 2\\
    croatia \(\rightarrow\) spain & 1 & 19 & 0 & 15 & 51 & 1 & 6 & 19 & 0\\
    croatia \(\rightarrow\) sweden & 1 & 19 & 1 & 19 & 61 & 8 & 10 & 19 & 3\\
    czechia \(\rightarrow\) france & 1 & 10 & 0 & 1 & 13 & 1 & 1 & 10 & 1\\
    czechia \(\rightarrow\) germany & 1 & 20 & 0 & 1 & 22 & 1 & 1 & 20 & 2\\
    czechia \(\rightarrow\) greece & 1 & 68 & 4 & 1 & 85 & 0 & 1 & 68 & 2\\
    czechia \(\rightarrow\) ireland & 1 & 4 & 0 & 1 & 7 & 1 & 1 & 4 & 2\\
    czechia \(\rightarrow\) italy & 1 & 22 & 1 & 1 & 25 & 1 & 1 & 22 & 0\\
    czechia \(\rightarrow\) netherlands & 1 & 9 & 2 & 1 & 12 & 1 & 1 & 9 & 2\\
    czechia \(\rightarrow\) poland & 1 & 1 & 0 & 1 & 2 & 1 & 1 & 1 & 0\\
    czechia \(\rightarrow\) spain & 1 & 15 & 0 & 1 & 16 & 1 & 1 & 15 & 0\\
    czechia \(\rightarrow\) sweden & 1 & 19 & 1 & 1 & 30 & 1 & 1 & 19 & 2\\
    france \(\rightarrow\) germany & 1 & 20 & 0 & 10 & 39 & 6 & 4 & 20 & 2\\
    france \(\rightarrow\) greece & 1 & 68 & 2 & 10 & 108 & 6 & 10 & 68 & 5\\
    france \(\rightarrow\) ireland & 1 & 10 & 0 & 4 & 20 & 3 & 3 & 10 & 2\\
    france \(\rightarrow\) italy & 1 & 22 & 1 & 10 & 52 & 6 & 10 & 25 & 0\\
    france \(\rightarrow\) netherlands & 1 & 10 & 1 & 9 & 25 & 4 & 6 & 11 & 2\\
    france \(\rightarrow\) poland & 1 & 10 & 0 & 1 & 11 & 12 & 1 & 10 & 2\\
    france \(\rightarrow\) spain & 1 & 15 & 0 & 10 & 38 & 5 & 4 & 16 & 0\\
    france \(\rightarrow\) sweden & 1 & 19 & 1 & 10 & 39 & 4 & 8 & 19 & 1\\
    germany \(\rightarrow\) greece & 1 & 68 & 4 & 20 & 104 & 30 & 9 & 68 & 2\\
    germany \(\rightarrow\) ireland & 1 & 20 & 0 & 4 & 31 & 6 & 4 & 20 & 3\\
    germany \(\rightarrow\) italy & 1 & 22 & 1 & 20 & 57 & 26 & 11 & 22 & 2\\
    germany \(\rightarrow\) netherlands & 1 & 20 & 1 & 9 & 43 & 14 & 6 & 20 & 3\\
    germany \(\rightarrow\) poland & 1 & 20 & 2 & 1 & 26 & 6 & 1 & 20 & 2\\
    germany \(\rightarrow\) spain & 1 & 20 & 0 & 15 & 39 & 16 & 2 & 20 & 0\\
    germany \(\rightarrow\) sweden & 1 & 20 & 1 & 19 & 48 & 18 & 11 & 20 & 5\\
    greece \(\rightarrow\) ireland & 1 & 68 & 4 & 4 & 79 & 12 & 3 & 68 & 4\\
    greece \(\rightarrow\) italy & 1 & 68 & 3 & 22 & 131 & 19 & 22 & 68 & 3\\
    greece \(\rightarrow\) netherlands & 1 & 68 & 6 & 9 & 85 & 21 & 7 & 68 & 5\\
    greece \(\rightarrow\) poland & 1 & 68 & 6 & 1 & 84 & 17 & 1 & 68 & 2\\
    greece \(\rightarrow\) spain & 1 & 68 & 0 & 15 & 111 & 11 & 14 & 68 & 3\\
    greece \(\rightarrow\) sweden & 1 & 68 & 3 & 19 & 125 & 11 & 12 & 68 & 2\\
    ireland \(\rightarrow\) italy & 1 & 22 & 1 & 4 & 36 & 4 & 4 & 22 & 2\\
    ireland \(\rightarrow\) netherlands & 1 & 9 & 2 & 4 & 21 & 4 & 3 & 10 & 3\\
    ireland \(\rightarrow\) poland & 1 & 4 & 0 & 1 & 7 & 8 & 1 & 4 & 2\\
    ireland \(\rightarrow\) spain & 1 & 15 & 0 & 4 & 25 & 4 & 3 & 15 & 0\\
    ireland \(\rightarrow\) sweden & 1 & 19 & 0 & 4 & 31 & 4 & 4 & 19 & 3\\
    italy \(\rightarrow\) netherlands & 1 & 22 & 2 & 9 & 42 & 15 & 9 & 22 & 2\\
    italy \(\rightarrow\) poland & 1 & 22 & 1 & 1 & 29 & 14 & 1 & 22 & 1\\
    italy \(\rightarrow\) spain & 1 & 22 & 0 & 15 & 51 & 12 & 12 & 22 & 0\\
    italy \(\rightarrow\) sweden & 1 & 22 & 1 & 19 & 68 & 13 & 13 & 22 & 0\\
    netherlands \(\rightarrow\) poland & 1 & 9 & 3 & 1 & 13 & 6 & 1 & 9 & 2\\
    netherlands \(\rightarrow\) spain & 1 & 15 & 0 & 9 & 29 & 5 & 5 & 15 & 2\\
    netherlands \(\rightarrow\) sweden & 1 & 19 & 2 & 9 & 39 & 5 & 6 & 19 & 2\\
    poland \(\rightarrow\) spain & 1 & 15 & 0 & 1 & 17 & 6 & 1 & 15 & 0\\
    poland \(\rightarrow\) sweden & 1 & 19 & 1 & 1 & 29 & 5 & 1 & 19 & 1\\
    spain \(\rightarrow\) sweden & 1 & 19 & 0 & 15 & 47 & 3 & 4 & 19 & 0\\
    wish \(\rightarrow\) luck & 1 & 5 & 2 & 4 & 44 & 5 & 4 & 22 & 0\\
    kick \(\rightarrow\) stuff & 1 & 5 & 3 & 5 & 29 & 6 & 5 & 18 & 0\\
    try \(\rightarrow\) it & 1 & 3 & 1 & 3 & 27 & 4 & 3 & 11 & 0\\
    f serif \(\rightarrow\) f sans & 1 & 1 & 0 & 1 & 1 & 3 & 1 & 1 & 0\\
    i serif \(\rightarrow\) i sans & 2 & 2 & 0 & 2 & 2 & 5 & 2 & 2 & 0\\
    u serif \(\rightarrow\) u sans & 1 & 1 & 0 & 1 & 1 & 2 & 1 & 1 & 0\\
\end{longtable}

\end{document}